\documentclass[journal,twocolumn]{IEEEtran}

\usepackage{amsmath,amssymb,amsfonts}
\usepackage{graphicx}
\usepackage{cite}
\usepackage{algorithmic}
\usepackage{algorithm}
\usepackage{url}
\usepackage{multirow}
\usepackage{booktabs}
\usepackage{mathtools}
\usepackage{enumitem}
\usepackage{placeins}

\newcommand{\bh}{\mathbf{h}}

\newcommand{\CC}{\mathbb{C}}

\usepackage{acronym}
\usepackage{bm}
\usepackage{bbm}
\usepackage{color, soul}
\usepackage{indentfirst}
\usepackage{lipsum}
\usepackage{setspace}
\usepackage{tikz}
\usetikzlibrary{arrows}
\usepackage{subfigure}
\usepackage[amsmath,thmmarks]{ntheorem}
\usepackage{theorem}

\def \sl {single-layer }

\newtheorem{theorem}{\bf Theorem}

\AtBeginDocument{
\setlength{\textfloatsep}{6pt plus 1pt minus 2pt}
\setlength{\floatsep}{5pt plus 1pt minus 2pt}
\setlength{\intextsep}{5pt plus 1pt minus 2pt}
\setlength{\abovecaptionskip}{2pt}
\setlength{\belowcaptionskip}{-2pt}
\setlength{\abovedisplayskip}{4pt plus 1pt minus 2pt}
\setlength{\belowdisplayskip}{4pt plus 1pt minus 2pt}
\setlength{\abovedisplayshortskip}{2pt plus 1pt minus 1pt}
\setlength{\belowdisplayshortskip}{2pt plus 1pt minus 1pt}
\setlength{\jot}{1pt}
}

\begin{document}
	
	\title{Blockage-Robust Beamforming for Near-Field Communications: From Single-Airy to Multi-Airy}
	
	\author{Yi~Wang and Linglong~Dai, {\textit{Fellow, IEEE}}
		\thanks{This work was supported in part by the National Science Fund for Distinguished Young Scholars under Grant 62325106, in part by the National Science and Technology Major Projects of China under Grant 2025ZD1301800, and in part by the National Key Research and Development Program of China under Grant 2023YFB3811503.}
		\thanks{The authors are with the Department of Electronic Engineering and the State Key Laboratory of Space Network and Communications, Tsinghua University, Beijing 100084, China (e-mails: yiwang24@mails.tsinghua.edu.cn, daill@tsinghua.edu.cn).
	}
	}
	
	\maketitle

\begin{abstract}
High-frequency communications strongly depend on the line-of-sight (LoS) path, and obstacle blockage can severely degrade the received signal power and achievable rate.
Near-field Airy beams with curved trajectories can circumvent obstacles, offering a promising way to alleviate blockage.
However, since an Airy beam carries most useful energy along a single curved trajectory, existing Airy beamforming methods are highly sensitive to estimation errors of transmitter-obstacle-receiver geometry. That is to say, even a small error in the estimated geometry may cause the mismatched Airy trajectory, leading to severe performance loss.
To address this problem, we propose a multi-Airy beamforming scheme for blockage-robust near-field communications.
Specifically, we first reveal and analyze the sensitivity mechanism of single-Airy beamforming.
This mechanism motivates us to extend the single-Airy generation method to a coordinated multi-Airy generation method by deriving the phase offsets required to coherently combine multiple Airy beams at the target user.
Based on this coordinated generation method, we partition the transmit array into multiple sub-arrays and configure a tailored Airy beam for each sub-array, so that the resulting Airy beams formed by multiple curved trajectories can be coherently combined at the target user.
Simulation results verify the sensitivity of single-Airy beamforming and the robustness of multi-Airy beamforming under estimation errors of transmitter-obstacle-receiver geometry.
Moreover, the proposed scheme achieves higher achievable rates than single-Airy beamforming in blocked scenarios without geometry estimation errors.
\end{abstract}
	
	
	\vspace{-1em}
	\section{Introduction} \label{sec-intro}
	
	6G wireless systems are evolving towards millimeter-wave and Terahertz bands to exploit wider bandwidths to improve the peak data rate. At these high frequencies, severe reflection and refraction losses make wireless links more dependent on line-of-sight (LoS) propagation. Thus, large antenna apertures are usually required to achieve high array gains through more directional transmission~\cite{2024tutorial,2021THz,optics,nearfield}. Since large apertures concentrate the transmitted signal energy on narrow directed beams, blockage of the LoS path by obstacles can cause severe degradation in received signal power and achievable rate, or even link outage~\cite{shurakov2023empirical,hanchong}.

	Several techniques have been investigated to alleviate blockage-induced high-frequency link degradation. For example, distributed base-station architectures can connect a user to multiple geographically separated base stations, thereby exploiting spatial diversity against blockage~\cite{maccartney2019bsdiversity,kumar2021blockageaware}. Such distributed architectures can improve link reliability at the cost of additional base-station deployment, inter-site coordination, and fronthaul/synchronization overhead. To avoid the need for additional active base stations, reconfigurable intelligent surfaces (RISs) have been studied as lower-cost passive auxiliary devices to create additional paths for blocked high-frequency links~\cite{2017recent,chen2021towards}. Unfortunately, RIS-aided solutions still require the deployment of additional programmable hardware. The practical application of distributed and RIS-aided architectures above can be limited by these additional hardware costs. Therefore, a more practical solution addressing the blockage problem without additional hardware is desirable.

	Fortunately, without introducing additional hardware, near-field Airy beams are attractive for addressing the blockage problem, since their curved trajectories can circumvent obstacles in the LoS path~\cite{overview,selfhealing1,selfhealing2,darsena2025airy}. Specifically, in contrast to directed beams that mainly concentrate signal energy along a direct LoS path, the self-accelerating property enables Airy beams to follow curved trajectories during propagation. Thus, they can deliver useful signal power to the target user after avoiding the obstacles. In addition, their approximately nondiffracting property helps to maintain a concentrated main lobe over distance, and their self-healing property enables the main lobe of the Airy beam to be reconstructed from the remaining side lobes after blockage. Moreover, from an implementation perspective, Airy beams can be directly generated by controlling the existing transmit array without additional hardware~\cite{hanchong,hanchong2,chenhz}. To sum up, Airy beams provide a promising solution to the blockage problem in high-frequency communication systems.

	\subsection{Prior Works}
	\vspace{-0.2em}
The Airy wave was first theoretically identified in 1979 as a nonspreading wave-packet solution in quantum mechanics~\cite{airy1}. Then its nonspreading and accelerating properties were studied in~\cite{airy2,airy3}. Since the paraxial diffraction equation in optics has exactly the same mathematical form as the free-particle Schr\"odinger equation in quantum mechanics, the Airy-wave solution was later introduced to the optics community~\cite{airy1,airy5}. The first experimental observation of Airy beams in optics was reported in 2007~\cite{airy5}, which marked the transition of Airy beams from a theoretical concept to an experimentally observed physical phenomenon. Subsequent studies verified the self-healing and approximately nondiffracting propagation properties of Airy beams~\cite{airy4,selfhealing1,selfhealing2,airy6}. In addition, the ideal Airy beam has infinite energy and thus is not physically realizable. To this end, finite-energy Airy beams were introduced to make Airy beam generation practical~\cite{airyfinite}.
	
These properties, observed and studied in optics, motivate the application of Airy beams to address the blockage problem in wireless communications. For example,~\cite{airy7} investigated the use of Airy beams for obstacle circumvention in wireless applications for the first time in 2022. Subsequent studies examined Airy beam by both theoretical evaluation and experimental validation. Theoretically, some works compared Airy beams with Gaussian, focused, and other curved-beam architectures under blockage~\cite{hanchong,songlingyang,li2026beammanipulation}. These works considered Airy beams as representative wavefront-engineering tools to address the blockage problem in wireless communications. Experimentally, several works had validated the blockage-circumvention capability of Airy beams in practical high-frequency communication systems~\cite{lee,curving}.

Based on these theoretical and experimental studies of Airy beams, researchers further investigated how to design communication-oriented Airy beams. Specifically,~\cite{hanchong2} showed that Airy beams can be generated by controlling phased arrays rather than by relying on conventional optical-lens hardware in optics, which improves the practicality of Airy beam deployment in wireless communication systems. Based on this generation method,~\cite{zhao2026efficienttraining} developed an efficient Airy beam training framework by codebook pruning and fast scanning. Unlike these conventional model-driven designs,~\cite{chenhz} proposed a data-driven learning framework to search for the optimal trajectory of Airy beams under blockage. In summary, these works demonstrate the feasibility of Airy beams for blockage mitigation in high-frequency communication systems.
	
However, the performance of Airy beams is highly sensitive to estimation errors of transmitter-obstacle-receiver geometry, i.e., even a small error in estimating this geometry can cause a significant loss of the received signal power. Specifically, this sensitivity mainly stems from the single-trajectory structure adopted by most existing Airy beam schemes, where a single dominant Airy main lobe carries most of the useful signal energy. Therefore, the Airy trajectory must be accurately selected according to the transmitter-obstacle-receiver geometry. If the estimated geometry deviates from the actual geometry, the obstacle may block this single trajectory and destroy the main lobe, so the signal energy received at the target user can drop significantly. For example, as shown later in this paper, a \(9.4\%\) geometry estimation error can result in a \(95.2\%\) loss of received signal energy. This finding motivates a robust Airy beamforming scheme alleviating the dependence on a single accurately selected curved trajectory.
	
\subsection{Our Contributions}
In this paper, we propose a multi-Airy beamforming scheme to address the single-trajectory sensitivity problem of existing Airy beam schemes. The proposed scheme partitions the transmit array into multiple sub-arrays and generates multiple coordinated Airy beams to mitigate the sensitivity caused by relying on a single curved trajectory.\footnote{Simulation codes will be provided to reproduce the results in this article: \url{http://oa.ee.tsinghua.edu.cn/dailinglong/publications/publications.html}.} Specifically, the main contributions of this paper are summarized as follows:
\begin{itemize}
		\item First, we reveal and analyze the serious sensitivity problem of single-Airy beamforming under estimation errors of transmitter-obstacle-receiver geometry. Our analysis shows that a small geometry estimation error can result in an obvious loss of the received signal power at the target user when the actual obstacle blocks the designed Airy trajectory and destroys the main lobe. For example, a \(9.4\%\) geometry estimation error can lead to a \(95.2\%\) loss of the received signal energy.
		\item Based on the analysis above, we then extend the single-Airy generation method into a coordinated multi-Airy generation method. This extension characterizes how multiple Airy beams can be generated and configured to be coherently combined at the target user. Since the derivation only requires phase alignment among independently generated Airy beams, it can also be applied to more general applications such as distributed systems.
		\item To address the sensitivity problem, we develop a multi-Airy beamforming scheme based on the coordinated multi-Airy generation method. To be specific, the proposed scheme partitions the transmit array into multiple sub-arrays, and each sub-array is configured a tailored Airy beam. These generated Airy beams with different curved trajectories can be coherently combined at the target user. The resulting trajectory diversity alleviates the sensitivity caused by relying on a single accurately selected Airy trajectory.
		\item Finally, simulation results verify that the proposed multi-Airy beamforming scheme is much more robust than the single-Airy beamforming scheme under estimation errors of transmitter-obstacle-receiver geometry. For example, under a 40-mm geometry error, the optimal multi-Airy configuration compensates for the achievable-rate loss, reducing it from 97.2\% with single-Airy beamforming to just 33.3\%. Moreover, under blocked propagation without geometry estimation errors, multi-Airy beamforming still achieves higher rates than single-Airy beamforming, with up to a \(2.57\times\) rate improvement in high-blockage cases.
\end{itemize}

	\subsection{Organization and Notation}
	\subsubsection{Organization}
	The rest of this paper is organized as follows. Section II presents the near-field channel and blockage propagation models. Section III briefly reviews the Airy beam propagation and generation principles. Section IV analyzes the limitations of existing single-Airy beamforming. Section V develops the coordinated multi-Airy generation method. Section VI presents the proposed multi-Airy beamforming scheme together with the robustness analysis. Section VII provides the simulation results, and Section VIII concludes the paper.
	
	\subsubsection{Notation}
	${\bf a}^H$ and ${\bf A}^H$ denote the conjugate transpose of vector ${\bf a}$ and matrix ${\bf A}$, respectively. The notation $\|{\bf a}\|_2$ denotes the Euclidean norm of ${\bf a}$, \(\operatorname{supp}(\cdot)\) denotes the support of a function, and \(\mathbbm{1}(\cdot)\) denotes the indicator function. The operator \(\mathcal{O}(\cdot)\) is used for complexity order.
	
\section{System Model} \label{sec-model}
We first characterize the near-field channel and then the blockage propagation model, laying the groundwork for the following sections.

	\subsection{Near-Field Channel Model}
	We consider a near-field communication system with a uniform linear array (ULA) transmitter (Tx) and a single-antenna receiver (Rx). The ULA is placed on the $x$-axis at $z=0$ and is centered at the origin. The $n$-th antenna coordinate is
	\begin{equation}
		x_n=\left(n-\frac{N+1}{2}\right)d,\quad n=1,2,\cdots,N,
	\end{equation}
	where $d\leq \lambda/2$ is adopted to suppress spatial aliasing and grating lobes, $\lambda=c/f_c$ is the wavelength, $c$ is the speed of light, and $f_c$ is the carrier frequency. The physical aperture length is denoted by $L=(N-1)d$. The Rx, which is also the intended focal target, is located at
	\begin{equation}
		(x_0,z_0)=(r\sin\theta,r\cos\theta)
	\end{equation}
	in the $x$-$z$ plane. This point is the nominal receiver location. The target window used later is an evaluation region around this location, representing finite positioning tolerance or a small receive aperture, rather than a change to the channel center. We focus on the radiating near-field regime, where the link distance is below the aperture-dependent Rayleigh distance, on the order of $2L^2/\lambda$, and the spherical phase variation across the aperture cannot be replaced by a plane-wave approximation.
	
	For the unblocked reference link, the line-of-sight (LoS) near-field channel vector $\bh \in \CC^{N \times 1}$ is
	\begin{equation}
		[\bh]_n = \frac{\rho}{L_n} e^{-\jmath k L_n}, \quad n = 1,2,\cdots,N,
		\label{eq:near_field_channel}
	\end{equation}
	where $\rho$ absorbs the large-scale attenuation and antenna constants, $k=2\pi/\lambda$ is the wavenumber, and
	\begin{equation}
		L_n=\sqrt{(x_0-x_n)^2+z_0^2}
	\end{equation}
	is the distance from the $n$-th antenna element to the Rx. In the near field, the exact distance $L_n$ must be retained. Its second-order expansion around the array center is
	\begin{equation}
		L_n \approx r - x_n\sin\theta + \frac{x_n^2\cos^2\theta}{2r},
	\end{equation}
	which contains both the linear steering phase and the quadratic focusing phase. This term is the basic reason that near-field beamforming must match a spherical wavefront rather than only a propagation angle.

\subsection{Blocked Propagation Model} \label{sec-blockage-model}
We consider single-sided blockage between the transmit aperture and the receiver. At high frequencies, many obstacles strongly attenuate penetration, while reflected or scattered components are typically weak at the target~\cite{hanchong}. We therefore adopt a single-sided screen model that retains free-space propagation through the unblocked region and diffraction of the truncated field. This screen, together with the transmitter and receiver locations, defines the transmitter-obstacle-receiver (Tx-obs-Rx) geometry used hereafter. More complex obstacle shapes can be approximated by finer screen elements and the corresponding propagation integrals.

\begin{figure}
	\centering
	\includegraphics[width=\columnwidth]{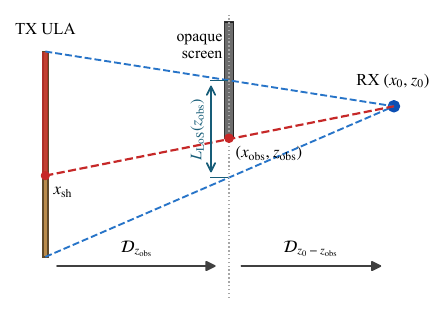}
	\caption{Representative planar-screen Tx-obs-Rx geometry with edge-error notation.}
	\label{fig:blocked_propagation_model}
\end{figure}

Fig.~\ref{fig:blocked_propagation_model} illustrates the blocked-propagation geometry. Under the coordinate system defined above, the transmit array lies on the \(x\)-axis at \(z=0\), and propagation is described in the \(x\)-\(z\) plane. For tractability, we consider a zero-thickness obstacle screen parallel to the array at \(z=z_{\mathrm{obs}}\). Under this geometry, propagation from \(z=0\) to \(z=z_0\) is modeled by two free-space diffraction segments separated by one obstacle mask.

\begin{figure*}
	\centering
	\includegraphics[width=0.85\textwidth]{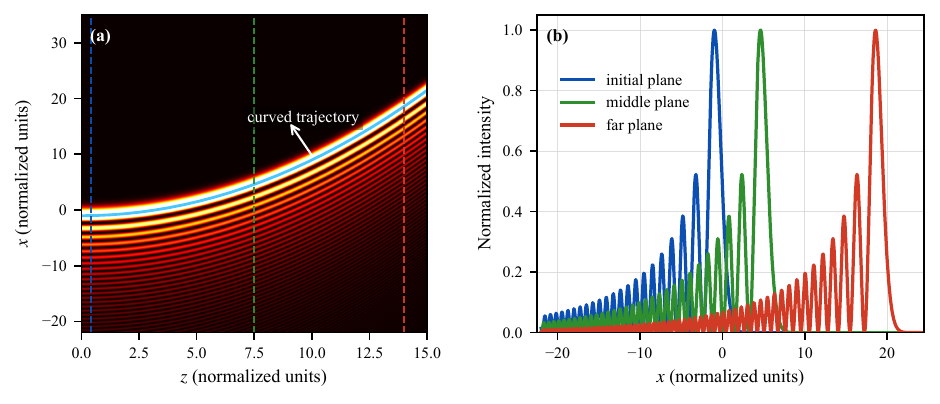}
	\caption{Conceptual illustration of Airy beam propagation: (a) curved trajectory during propagation; (b) approximately preserved transverse intensity profiles at different distances.}
	\label{fig:airy_beam_properties}
\end{figure*}

For a physical obstacle edge located at \((x_{\mathrm{obs}},z_{\mathrm{obs}})\), the corresponding shadow boundary projected onto the transmit aperture is
\begin{equation} \label{eq:blockage_window}
	x_{\mathrm{sh}}
	=
	\frac{x_{\mathrm{obs}}z_0-x_0z_{\mathrm{obs}}}
	{z_0-z_{\mathrm{obs}}},
	\qquad
	\mathrm{bl}=\frac{L/2-x_{\mathrm{sh}}}{L},
\end{equation}
where \(\mathrm{bl}\) is the invisible ratio of the transmit aperture from the target under the single-sided shadow approximation. For \(0\leq \mathrm{bl}\leq1\), it equals the invisible aperture fraction; values \(\mathrm{bl}>1\) are used later only as a geometric extrapolation beyond complete target-visible aperture blockage. This definition separates the aperture shadow from the longitudinal obstacle position, so a fixed \(\mathrm{bl}\) represents the same aperture shadow when \(z_{\mathrm{obs}}\) changes.
At \(z=z_{\mathrm{obs}}\), the aperture-endpoint-to-target rays intersect the obstacle plane at two points. Their transverse distance defines the effective LoS width
\begin{equation}
	L_{\mathrm{LoS}}(z_{\mathrm{obs}})
	=
	L\frac{z_0-z_{\mathrm{obs}}}{z_0}.
\end{equation}
Let \(\hat{x}_{\mathrm{obs}}\) denote the obstacle-edge coordinate inferred from the estimated Tx-obs-Rx geometry. In the planar-screen model, the signed geometry estimation error and its normalized magnitude are parameterized as
\begin{equation} \label{eq:normalized_edge_error}
	\Delta x_e
	\triangleq
	\hat{x}_{\mathrm{obs}}-x_{\mathrm{obs}},
	\qquad
	\epsilon_e
	\triangleq
	\frac{|\Delta x_e|}
	{L_{\mathrm{LoS}}(z_{\mathrm{obs}})}.
\end{equation}
Thus, \(\Delta x_e\) is the signed obstacle-edge coordinate error induced by the Tx-obs-Rx geometry estimation error, and \(\epsilon_e\) is its magnitude normalized by the effective LoS width. Since geometry errors mainly perturb the effective obstacle edge in this planar model, their sensitivity impact mainly enters through \(\Delta x_e\). This connection is illustrated in Section \ref{subsec:single_airy_sensitivity} and analyzed in Section \ref{subsec:robustness_analysis}. Therefore, unless otherwise stated, the Tx-obs-Rx geometry error and its induced obstacle-edge error are used interchangeably as a local approximation in this paper.
The corresponding hard truncation mask at the obstacle plane is
\begin{equation}
	M_{\mathrm{obs}}(x;x_{\mathrm{obs}})
	=
	\begin{cases}
		1, & x<x_{\mathrm{obs}},\\
		0, & x\ge x_{\mathrm{obs}}.
	\end{cases}
\end{equation}

Let \(\mathcal{D}_{z}\) denote the scalar free-space propagation operator over distance \(z\), defined by the Rayleigh-Sommerfeld diffraction integral~\cite{hanchong2}. For an aperture field \(\psi_0(x)\), the incident field at the obstacle plane is
\begin{equation} \label{eq:pre_blockage_prop}
	\psi(x,z_{\mathrm{obs}})
	=
	\mathcal{D}_{z_{\mathrm{obs}}}\{\psi_0\}(x).
\end{equation}
Thus, the field after truncation is
\begin{equation} \label{eq:blockage_truncation}
	\psi'(x,z_{\mathrm{obs}})
	=
	\psi(x,z_{\mathrm{obs}})
	M_{\mathrm{obs}}(x;x_{\mathrm{obs}}).
\end{equation}
Propagation over the remaining distance gives the received field at the target plane:
\begin{equation} \label{eq:post_blockage_prop}
	\psi_{\mathrm{tot}}(x,z_0)
	=
	\mathcal{D}_{z_0-z_{\mathrm{obs}}}
	\{\psi'(\cdot,z_{\mathrm{obs}})\}(x).
\end{equation}
For an arbitrary obstacle-edge coordinate \(x_e\), the blocked propagation operator is
\begin{equation} \label{eq:blockage_diffraction_operator}
	\begin{aligned}
		\mathcal{B}_{x_e,z_{\mathrm{obs}}}
		\{\psi_0\}(x)
		&\triangleq
		\mathcal{D}_{z_0-z_{\mathrm{obs}}}
		\left\{
		M_{\mathrm{obs}}(\cdot;x_e)
		\right.\\[-0.2em]
		&\qquad\left.
		{}\times\mathcal{D}_{z_{\mathrm{obs}}}\{\psi_0\}(\cdot)
		\right\}(x).
	\end{aligned}
\end{equation}
For \(x_e=x_{\mathrm{obs}}\), \(\psi_{\mathrm{tot}}(x,z_0)=\mathcal{B}_{x_{\mathrm{obs}},z_{\mathrm{obs}}}\{\psi_0\}(x)\). With this operator, the near-field reference channel and the blocked propagation model serve as the basis for the following analysis.

\section{Fundamentals of Airy Beam} \label{sec-airy-fundamentals}
This section reviews the Airy beam theory and generation principles. We first introduce the fundamental propagation properties and mathematical representation of Airy beams. We then explain how Airy beams can be generated directly by antenna arrays.

\subsection{Propagation of Airy Beam}

Airy beams are accelerating wavepackets governed by the paraxial Helmholtz equation. They have recently been revisited for wireless design because their curved trajectories can be exploited for practical blockage-mitigation tasks~\cite{overview,hanchong2,darsena2025airy}. Airy beams exhibit self-accelerating propagation, approximate nondiffraction, and self-healing after blockage. Fig.~\ref{fig:airy_beam_properties} illustrates the first two properties: the beam follows a curved trajectory in Fig.~\ref{fig:airy_beam_properties}(a), and its normalized transverse intensity profile remains approximately preserved during propagation in Fig.~\ref{fig:airy_beam_properties}(b). These properties make Airy beams attractive for blockage-robust communication.

The scalar field of a 1D Airy beam in free space is governed by the scalar paraxial wave equation:
\begin{equation}
	\left( \frac{\partial^2}{\partial x^2} + 2\jmath k \frac{\partial}{\partial z} \right) \psi(x,z) = 0,
\end{equation}
where $\psi(x,z)$ denotes the complex field at transverse coordinate $x$ and propagation distance $z$, $k=2\pi/\lambda$ is the free-space wavenumber, and $\jmath$ is the imaginary unit. For an ideal Airy beam, the input-plane boundary condition at $z=0$ is $\psi(x,0)=\text{Ai}(\gamma x)$,
where $\gamma$ is the transverse scaling factor and is inversely related to the main lobe width.

Under this boundary condition, the closed-form solution is
\begin{equation}
	\psi(x,z) = \text{Ai}\left(\gamma x - \frac{\gamma^4 z^2}{4k^2} \right) e^{\jmath \left( \frac{\gamma^2 z}{2k} \right) \left( \gamma x - \frac{\gamma^4 z^2}{6 k^2} \right)},
\end{equation}
where $\text{Ai}(\cdot)$ is the Airy function of the first kind. The Airy term determines the transverse intensity envelope, and the exponential term provides the associated phase distribution. The high-intensity caustic corresponds to a constant Airy-function argument. Ignoring this constant transverse offset gives the parabolic Airy trajectory
\begin{equation}
	x(z) = \frac{\gamma^3}{4k^2} z^2,
\end{equation}
which describes the deterministic parabolic bending of the Airy trajectory during free-space propagation. In the ideal solution, the sign of $\gamma$ determines the bending direction, and its magnitude controls the bending strength.

Because the ideal Airy beam has infinite total energy, practical transmission uses a finite-energy version obtained by exponential truncation:
\begin{equation}
	\psi(x,0) = \operatorname{Ai}\left( \gamma x \right) e^{\alpha_{\mathrm A}\gamma x},
	\label{eq:finite_energy_airy}
\end{equation}
where $\alpha_{\mathrm A}>0$ controls the Airy truncation in the normalized transverse coordinate, and a larger $\alpha_{\mathrm A}$ concentrates more power in the main lobe~\cite{overview}. In \eqref{eq:finite_energy_airy}, \(\gamma>0\) is used for definiteness so that the exponential apodization suppresses the oscillatory tail. The opposite bending direction can be represented by the reflected coordinate \(x\mapsto -x\), and is later controlled directly by the trajectory parameter \(B\) in the array-based formulation.

\subsection{Array-Based Generation of Airy Beam} \label{sec-equivalence}

In optics, Airy beam generation commonly relies on lens-based implementations, but such hardware is not necessarily required in a communication system. Phased arrays can generate a carefully designed initial amplitude and phase distribution and use the free-space propagation integral to replace the lens-based Fourier transform for Airy beam generation~\cite{hanchong,hanchong2,chenhz}. As shown in Fig.~\ref{fig:array_based_airy_generation_geometry}, the Airy beam is specified by three trajectory parameters \((B,F,\theta)\). The dashed line denotes the Airy generation plane, and the orange curve illustrates the Airy trajectory associated with these parameters in a representative Tx-obs-Rx geometry. The parameter \(B\), with unit \(\mathrm{m}^{-1}\), controls the bending strength and direction. The distance \(F\) specifies the quadratic focusing term associated with the Airy beam generation plane, and \(\theta\) is the steering angle measured from the positive \(z\)-axis.

\begin{figure}
	\centering
	\includegraphics[width=\columnwidth]{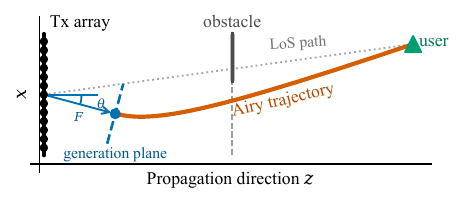}
	\caption{Array-based Airy beam generation geometry.}
	\label{fig:array_based_airy_generation_geometry}
\end{figure}

Following~\cite{hanchong2}, with Gaussian aperture width \(\omega_0\), the continuous aperture field for one Airy beam can be written as
\begin{equation}
	\begin{aligned}
	\psi_0(x')
	&=
	\exp\!\left[
	-\frac{x'^2}{\omega_0^2}
	+\jmath\Phi_{\mathrm A}(x';B,F,\theta)
	\right],\\
	\Phi_{\mathrm A}(x';B,F,\theta)
	&=
	\frac{(2\pi B)^3}{3}x'^3
	-\frac{\pi}{\lambda F}x'^2
	-\frac{2\pi}{\lambda}\sin\theta\,x'.
	\end{aligned}
	\label{eq:single_array_airy_field}
\end{equation}
The Gaussian factor realizes finite-energy truncation, and the cubic, quadratic, and linear phase terms determine Airy bending, near-field focusing, and steering, respectively.

Under the Fresnel/paraxial approximation and the continuous-aperture integral model, the corresponding closed-form Airy field is~\cite{hanchong2}
\begin{equation}
	\Psi_{\mathrm A}(x,z;B,F,\theta)
	=
	\frac{e^{\jmath kz}}{\jmath\lambda zB}
	e^{\jmath\pi x^2/(\lambda z)}
	e^{\jmath\phi_{\mathrm c}}
	\operatorname{Ai}(\xi),
	\label{eq:closed_form_airy_field}
\end{equation}
where
\begin{equation}
	\phi_{\mathrm c}
	=
	\frac{2C_2^3}{3A^2}
	-
	\frac{C_1C_2}{A},
	\quad
	A=(2\pi B)^3,
	\label{eq:closed_form_airy_phase_terms}
\end{equation}
\begin{equation}
	\begin{aligned}
	C_1
	&=
	-\left(
	\frac{2\pi}{\lambda}\sin\theta
	+
	\frac{2\pi x}{\lambda z}
	\right),\\
	C_2
	&=
	\frac{\pi}{\lambda}
	\left(
	\frac{1}{z}
	-
	\frac{1}{\widetilde F}
	\right),
	\quad
	\frac{1}{\widetilde F}
	=
	\frac{1}{F}
	-\jmath\frac{\lambda}{\pi\omega_0^2},
	\end{aligned}
	\label{eq:closed_form_airy_c_terms}
\end{equation}
and
\begin{equation}
	\xi(x,z;B,F,\theta)
	=
	-\frac{\sin\theta}{\lambda B}
	-\frac{x}{\lambda zB}
	-\frac{(1/z-1/\widetilde F)^2}
	{16\lambda^2\pi^2B^4}.
	\label{eq:closed_form_airy_argument}
\end{equation}
In this paper, \(\Psi_{\mathrm A}\) is used as the ideal single-Airy response for phase alignment.
For a discrete ULA, the continuous aperture field in \eqref{eq:single_array_airy_field} is sampled at the antenna locations \(x_n\):
\begin{equation}
	\psi_{0,n}(B,F,\theta)
	=
	\exp\!\left[
	-\frac{x_n^2}{\omega_0^2}
	+\jmath\Phi_{\mathrm A}(x_n;B,F,\theta)
	\right].
	\label{eq:single_array_airy_discrete}
\end{equation}

These Airy beam properties, together with the array-based generation principle, provide the foundation for the following multi-Airy construction.

\section{From Single-Airy to Multi-Airy} \label{sec-multi-airy-theory}
This section first identifies the sensitivity of single-Airy beamforming caused by its single-trajectory structure. Based on this observation, the single-Airy generation method is extended to a coordinated multi-Airy generation method. The phase-alignment rule for coherent combining is then derived, and a practical refinement is further provided under sub-array truncation and finite-window reception.

\subsection{Sensitivity of Single-Airy} \label{subsec:single_airy_sensitivity}

\begin{figure}
	\centering
	\includegraphics[width=0.90\linewidth]{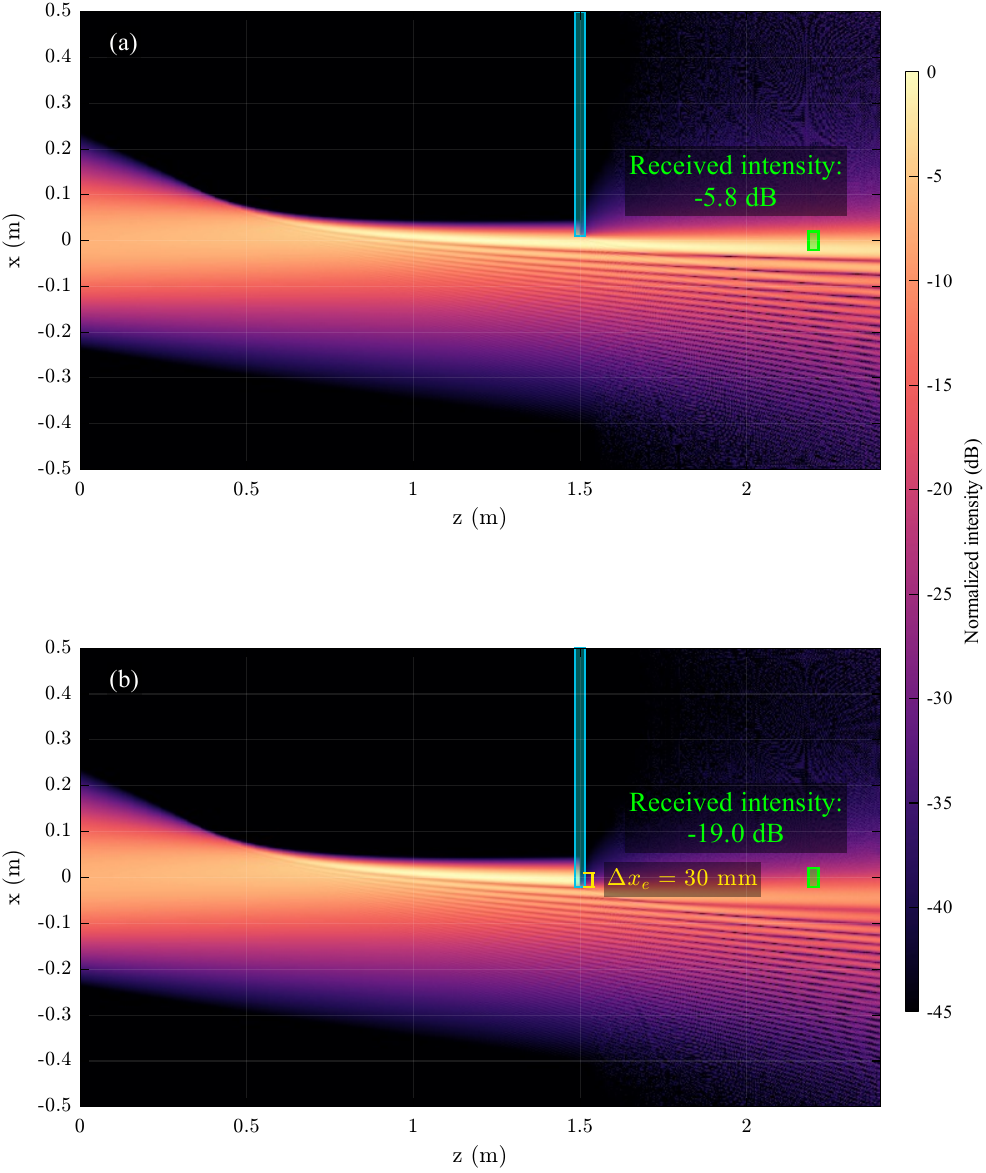}
	\caption{Single-Airy intensity distributions under Tx-obs-Rx geometry estimates: (a) accurate estimate; (b) inaccurate estimate with \(\Delta x_e=30\) mm.}
	\label{fig:single_airy_edge_error_field}
\end{figure}

Most received signal power of a single Airy beam is carried by the main lobe. To address blockage, the trajectory associated with this main lobe must be placed accurately with respect to the actual Tx-obs-Rx geometry. The remaining side lobes and diffracted field components are usually much weaker than the main lobe, so they cannot provide a comparably strong alternative trajectory when the main lobe is intercepted. Consequently, a small Tx-obs-Rx geometry estimation error can cause a sharp reduction in the received target signal power and achievable rate.

Fig.~\ref{fig:single_airy_edge_error_field} directly visualizes the resulting performance loss by comparing an accurate Tx-obs-Rx geometry estimate with a 30-mm geometry estimation error for the same target. In this example, \(L_{\mathrm{LoS}}(z_{\mathrm{obs}})=318\) mm, so the normalized geometry estimation error is only \(\epsilon_e\approx9.4\%\). This absolute error is also far below the 0.75-m horizontal positioning-accuracy requirement specified for the Indoor Factory-ISAC test environment in IMT-2030~\cite{ma2026imt2030tpr}. Nevertheless, the target-window average intensity decreases from \(-5.8\) dB to \(-19.0\) dB, corresponding to an approximately \(95.2\%\) loss of received signal energy in the linear scale. This comparison clearly demonstrates that a small Tx-obs-Rx geometry estimation error can produce a disproportionate reduction in the received target signal energy.

The discussion above serves as a qualitative explanation and an intuitive numerical example, which makes the subsequent comparison with the multi-Airy construction clearer. Detailed analysis is deferred to Section \ref{subsec:robustness_analysis}. There, a bound on the target-window energy variation under Tx-obs-Rx geometry estimation errors is derived, and the field-distribution conditions that lead to severe performance degradation are identified.

\subsection{Phase Alignment for Coordinated Multi-Airy Combining} \label{subsec:closed_form_phase_alignment}
Consider \(M\) independently generated Airy beams with spatially arbitrary reference centers \((x_{c,m},z_{c,m})\) and trajectory parameters \((B_m,F_m,\theta_m)\). For the ULA considered in this paper, this general description reduces to \(z_{c,m}=0\), with \(x_{c,m}\) being the sub-array center. The response of the \(m\)-th beam at the target point \((x_0,z_0)\) is
\begin{equation}
	\begin{aligned}
	g_m
	&=
	\Psi_{\mathrm A}(\Delta x_m,\Delta z_m;B_m,F_m,\theta_m)\\
	&=
	\frac{e^{\jmath k\Delta z_m}}{\jmath\lambda\Delta z_m B_m}
	e^{\jmath\pi\Delta x_m^2/(\lambda\Delta z_m)}
	e^{\jmath\phi_{\mathrm c,m}}
	\operatorname{Ai}(\xi_m),
	\end{aligned}
	\label{eq:closed_form_subarray_response}
\end{equation}
where
\(\Delta x_m=x_0-x_{c,m}\), \(\Delta z_m=z_0-z_{c,m}\), and
\begin{equation}
	(\phi_{\mathrm c,m},\xi_m)
	=
	(\phi_{\mathrm c},\xi)
	\big|_{(x,z,B,F,\theta)=(\Delta x_m,\Delta z_m,B_m,F_m,\theta_m)}.
	\label{eq:subarray_closed_form_terms}
\end{equation}
Therefore, the phase of the \(m\)-th target response can be written as
\begin{equation}
	\begin{aligned}
	\varphi_m
	&=\arg(g_m)\\
	&=
	k\Delta z_m
	+\frac{\pi\Delta x_m^2}{\lambda\Delta z_m}
	+\Re\{\phi_{\mathrm c,m}\}\\
	&\quad
	+\arg\!\left[\operatorname{Ai}(\xi_m)\right]
	-\frac{\pi}{2}
	-\arg(B_m).
	\end{aligned}
	\label{eq:subarray_target_response_phase}
\end{equation}
Here \(B_m\in\mathbb{R}\setminus\{0\}\), and \(\arg(\cdot)\) uses the principal branch. Different branch choices only add integer multiples of \(2\pi\) to \(\varphi_m\) and do not change the relative phase-alignment rule.
Assigning a constant phase offset \(\delta_m\) to the \(m\)-th Airy beam does not change its trajectory or aperture-power distribution. The coherent point response is therefore
\begin{equation}
	\Psi_{\mathrm{sum}}(x_0,z_0;\boldsymbol{\delta})
	=
	\sum_{m=1}^{M}e^{\jmath\delta_m}g_m,
	\label{eq:closed_form_multi_airy_sum}
\end{equation}
where \(\boldsymbol{\delta}=[\delta_1,\cdots,\delta_M]^T\). Maximizing the point-receiver power gives
\begin{equation}
	\max_{\delta_1,\cdots,\delta_M}
	\left|
	\sum_{m=1}^{M}e^{\jmath\delta_m}g_m
	\right|^2.
	\label{eq:point_phase_alignment_problem}
\end{equation}

\begin{theorem}[Phase Alignment Condition]
\label{thm:phase_alignment_condition}
	For the specified Airy-beam responses \(\{g_m\}_{m=1}^{M}\) at the target point, the target-response power in \eqref{eq:point_phase_alignment_problem} is maximized when all nonzero terms \(e^{\jmath\delta_m}g_m\) have the same phase. The maximum target-response power is \(\left(\sum_{m=1}^{M}|g_m|\right)^2\). If \(g_m\neq0\), one set of phase offsets that achieves this maximum is
	\begin{equation}
		\delta_m^{(0)}
		=
		\phi_0-\varphi_m,
		\quad m=1,\cdots,M,
		\label{eq:closed_form_phase_alignment_solution}
	\end{equation}
	where \(\phi_0\) is an arbitrary common phase.
\end{theorem}

The proof is given in Appendix~\ref{app:phase_alignment_proof}.

Equation \eqref{eq:closed_form_phase_alignment_solution} is the theoretical coherent-combining condition for multi-Airy beamforming: once the individual Airy beams are specified, only constant phase offsets are needed to coherently combine them at the target.

\begin{figure}
	\centering
	\includegraphics[width=0.90\linewidth]{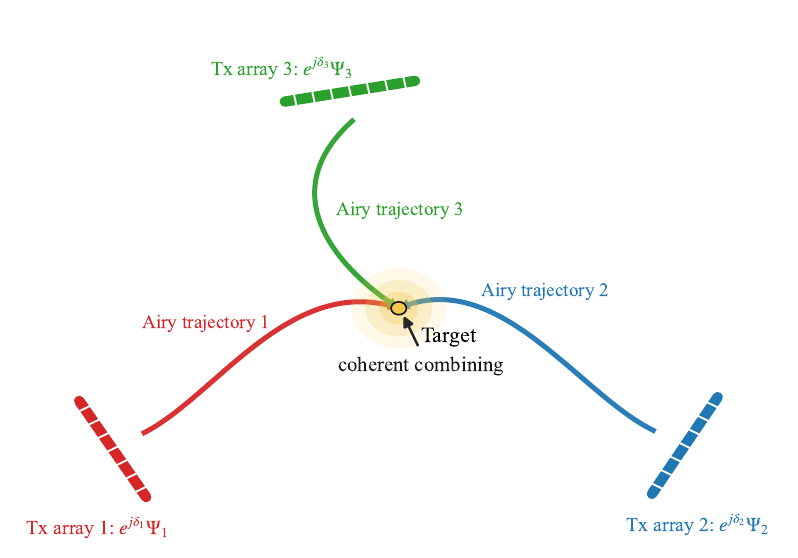}
	\caption{General coordinated multi-Airy combining from spatially separated arrays.}
	\label{fig:general_multi_airy_schematic}
\end{figure}
Fig.~\ref{fig:general_multi_airy_schematic} illustrates the general coordinated multi-Airy generation method, which can be applied to both centralized and distributed systems. As shown in the figure, multiple individual Airy beams can be generated from spatially separated or nonparallel transmit arrays, with their constant phase offsets calculated for coherent combining at the target user. Furthermore, Section \ref{sec-multi-airy} presents a specific sub-array-based implementation of this general method.

\subsection{Practical Refinement for Phase Alignment} \label{subsec:truncated_phase_refinement}
The phase-alignment rule above does not include the finite support of practical sub-arrays, which are truncated by their sub-array boundaries. The rule also aligns the fields at one point, but practical reception often considers the average energy over a target window rather than a single point. The target window can model a small receive aperture or a tolerance region around the nominal receiver location. These practical constraints require a refinement of the phase-alignment rule. Let \(\widetilde{\psi}_{0,m}(x)\) denote the actual aperture field of the \(m\)-th truncated sub-array, and define the unblocked propagated field over the target window
\begin{equation}
	u_m(x)
	=
	\left[\mathcal{D}_{z_0}\{\widetilde{\psi}_{0,m}\}\right](x),
	\qquad x\in\mathcal{R},
	\label{eq:actual_subarray_window_field}
\end{equation}
where \(\mathcal{D}_{z_0}\) is the free-space propagation operator defined in Section \ref{sec-blockage-model}, and
\begin{equation}
	\mathcal{R}
	=
	[x_0-W_{\mathrm{rx}}/2,\,
	x_0+W_{\mathrm{rx}}/2].
\end{equation}
The finite-window coherent-combining objective is
\begin{equation}
	J(\boldsymbol{\delta})
	=
	\frac{1}{W_{\mathrm{rx}}}
	\int_{\mathcal{R}}
	\left|
	\sum_{m=1}^{M}e^{\jmath\delta_m}u_m(x)
	\right|^2
	\mathrm{d}x
	=
	\mathbf{v}^{H}\mathbf{Q}\mathbf{v},
	\label{eq:window_phase_sync}
\end{equation}
where \(\mathbf{v}=[e^{\jmath\delta_1},\cdots,e^{\jmath\delta_M}]^T\), \(|v_m|=1\), and
\begin{equation}
	[\mathbf{Q}]_{mn}
	=
	\frac{1}{W_{\mathrm{rx}}}
	\int_{\mathcal{R}}u_m^*(x)u_n(x)\,\mathrm{d}x.
	\label{eq:window_phase_gram}
\end{equation}
For \(M=2\), the optimal relative phase is
\begin{equation}
	\delta_2^\star-\delta_1^\star
	=
	-\arg
	\int_{\mathcal{R}}u_1^*(x)u_2(x)\,\mathrm{d}x.
	\label{eq:two_beam_window_phase}
\end{equation}
For \(M>2\), \eqref{eq:window_phase_sync} is generally nonconvex. A feasible refinement can be obtained by cyclic coordinate ascent. Let \(\mathbf{q}_{\max}\) be the principal eigenvector of \(\mathbf{Q}\) and initialize \(v_m^{(0)}=e^{\jmath\arg([\mathbf{q}_{\max}]_m)}\). During sweep \(t+1\), define
\begin{equation}
	\begin{aligned}
		c_m^{(t+1)}
		&=
		\sum_{n<m}[\mathbf{Q}]_{mn}v_n^{(t+1)}
		+
		\sum_{n>m}[\mathbf{Q}]_{mn}v_n^{(t)},\\
		v_m^{(t+1)}
		&=
		\begin{cases}
			c_m^{(t+1)}/|c_m^{(t+1)}|,
			& c_m^{(t+1)}\neq0,\\
			v_m^{(t)}, & c_m^{(t+1)}=0.
		\end{cases}
	\end{aligned}
	\label{eq:window_phase_coordinate_update}
\end{equation}
With all other entries fixed, this update maximizes \(\mathbf{v}^{H}\mathbf{Q}\mathbf{v}\) with respect to \(v_m\), and the objective sequence is nondecreasing and bounded. The optimized offsets are recovered as \(\delta_m=\arg(v_m)\). The phase-alignment rule and its practical refinement provide the theoretical basis for proposed coordinated multi-Airy beamforming in section \ref{sec-multi-airy}.

\section{Proposed Multi-Airy Beamforming} \label{sec-multi-airy}
Based on the coordinated multi-Airy generation method in Section \ref{sec-multi-airy-theory}, this section presents the multi-Airy beamforming scheme. It then derives a bound on the target-window energy variation caused by Tx-obs-Rx geometry estimation errors. The bound further explains how trajectory diversity improves the robustness of multi-Airy beamforming.

	\subsection{Multi-Airy Beamforming Scheme}
	\subsubsection{Problem Formulation}
	To reduce the sensitivity caused by the single-trajectory structure of conventional Airy beamforming, we propose a beamforming scheme in which divided sub-arrays generate multiple coordinated Airy beams. Specifically, the full ULA is divided into \(M\) non-overlapping contiguous sub-arrays \(\{\mathcal{X}_m\}_{m=1}^{M}\). The sub-array partition can be chosen flexibly according to implementation constraints. For notational simplicity, unless otherwise stated, we use a uniform contiguous partition and assume that \(M\) divides \(N\). The \(m\)-th sub-array occupies
	\begin{equation}
		\mathcal{X}_m=[x_{m,\min},x_{m,\max}),
	\end{equation}
	where
	\begin{equation}
		x_{m,\min}=\frac{(m-1)L}{M}-\frac{L}{2},
		\qquad
		x_{m,\max}=\frac{mL}{M}-\frac{L}{2}.
	\end{equation}
	Its center and support function are
	\begin{equation}
		x_{c,m}=\frac{x_{m,\min}+x_{m,\max}}{2},
		\qquad
		\Pi_m(x)=\mathbbm{1}(x\in\mathcal{X}_m).
	\end{equation}
	The corresponding antenna-index set is
	\begin{equation}
		\mathcal{N}_m
		=
		\{n\in\{1,\cdots,N\}:x_n\in\mathcal{X}_m\}.
	\end{equation}
	Each sub-array is configured with a tailored Airy beam whose trajectory is independently specified. Following Section \ref{sec-multi-airy-theory}, the \(m\)-th trajectory is parameterized by
	\begin{equation}
		\boldsymbol{\eta}_m=(B_m,F_m,\theta_m),
	\end{equation}
	and \(\boldsymbol{\eta}=\{\boldsymbol{\eta}_m\}_{m=1}^{M}\) collects all sub-array trajectory parameters. The constant phase offsets \(\{\delta_m\}_{m=1}^{M}\) coordinate these beams according to Section \ref{sec-multi-airy-theory}. Let \(\psi_{0,m}(x;\boldsymbol{\eta}_m)\) denote the aperture field supported on \(\mathcal{X}_m\). The combined aperture field is
	\begin{equation}
		\psi_0(x;\boldsymbol{\eta},\boldsymbol{\delta})
		=
		\sum_{m=1}^{M}
		e^{\jmath\delta_m}\psi_{0,m}(x;\boldsymbol{\eta}_m),
	\end{equation}
	In this formulation, \(x_e\) denotes the estimated obstacle edge associated with the Tx-obs-Rx geometry used for trajectory selection. Under this estimated obstacle edge, the average received signal energy over the target window is
	\begin{equation}
		J_{\mathrm{B}}(\boldsymbol{\eta},\boldsymbol{\delta};x_e)
		=
		\frac{1}{W_{\mathrm{rx}}}
		\left\|
		\mathcal{P}_{\mathcal{R}}
		\mathcal{B}_{x_e,z_{\mathrm{obs}}}
		\{\psi_0(\cdot;\boldsymbol{\eta},\boldsymbol{\delta})\}
		\right\|_2^2,
		\label{eq:blocked_multi_airy_energy}
	\end{equation}
	where \(\mathcal{R}=[x_0-W_{\mathrm{rx}}/2,x_0+W_{\mathrm{rx}}/2]\) and \(\mathcal{P}_{\mathcal{R}}\) restricts a target-plane field to this receiver window.
	The multi-Airy beamforming objective can therefore be stated as
	\begin{equation}
		\begin{aligned}
			\underset{\boldsymbol{\eta},\boldsymbol{\delta}}{\operatorname{maximize}}
			\quad&
			J_{\mathrm{B}}(\boldsymbol{\eta},\boldsymbol{\delta};x_e)\\
			\operatorname{subject\ to}
			\quad&
			\operatorname{supp}\!\left(\psi_{0,m}\right)
			\subseteq\mathcal{X}_m,\quad m=1,\cdots,M,\\
			&
			\boldsymbol{\eta}_m\in\mathcal{T},
			\quad m=1,\cdots,M,\\
			&
			\sum_{n=1}^{N}|\psi_0(x_n;\boldsymbol{\eta},\boldsymbol{\delta})|^2\leq P_{\mathrm{tot}},
		\end{aligned}
		\label{eq:multi_airy_design_problem}
	\end{equation}
	where \(\mathcal{T}\) is the admissible set of Airy trajectory parameters, determined by the aperture size, maximum feasible phase gradient, paraxial validity range, and candidate trajectory constraints used by the beam-training or trajectory-selection stage. The first constraint requires each sub-array to be configured with a tailored Airy beam, the second keeps all trajectories within the admissible trajectory set, and the third enforces a common total transmit-power budget. Equation \eqref{eq:multi_airy_design_problem} states an ideal reference objective for blocked propagation under the estimated obstacle edge. Directly solving this problem is difficult because it requires the blocked propagation operator. The multi-Airy generation method in Section \ref{sec-multi-airy-theory} therefore enables a constructive beamforming rule for given trajectory parameters, based on sub-array Airy generation and coherent combining rather than full-array waveform optimization. Algorithm \ref{alg:constructive_multi_airy} does not solve \eqref{eq:multi_airy_design_problem} directly. Instead, it assumes that \(\boldsymbol{\eta}\) has been obtained by a separate beam-training or trajectory-selection stage and then constructs the corresponding multi-Airy beamforming vector.

	\subsubsection{Constructive Beamforming Procedure}
	Following the array-based Airy generation principle in Section \ref{sec-equivalence}, each sub-array samples the Airy aperture field associated with its trajectory parameters. The parameter \(\omega_0\) denotes the fixed Gaussian aperture width used in the Airy aperture field. Let
	\begin{equation}
		\bar{x}_{m,n}=x_n-x_{c,m}
	\end{equation}
	be the local coordinate of antenna \(n\) with respect to the \(m\)-th sub-array center. For \(n\in\mathcal{N}_m\), the corresponding unaligned beamforming weight is
	\begin{equation}
		[\widetilde{\mathbf{w}}_m]_n
		=
		\exp\!\left[
		-\frac{\bar{x}_{m,n}^2}{\omega_0^2}
		+\jmath\Phi_{\mathrm A}(\bar{x}_{m,n};B_m,F_m,\theta_m)
		\right].
		\label{eq:subarray_beamforming_field}
	\end{equation}
	For \(n\notin\mathcal{N}_m\), \([\widetilde{\mathbf{w}}_m]_n=0\).
	Equation \eqref{eq:subarray_beamforming_field} is the sub-array-restricted form of the discrete single-Airy weights in \eqref{eq:single_array_airy_discrete}. The sub-array number \(M\) and trajectory-parameter set \(\boldsymbol{\eta}=\{(B_m,F_m,\theta_m)\}_{m=1}^{M}\) are treated as specified inputs in the beamforming construction. This separates trajectory search from the proposed coherent multi-Airy construction. In practice, for a specified \(M\), received-signal-power-based Airy beam training, such as the method in~\cite{zhao2026efficienttraining}, can be used to probe candidate trajectory-parameter sets before applying the proposed multi-Airy beamforming rule. The receiver window \(\mathcal{R}\) is used only for the finite-window phase refinement below.

	The independently generated sub-array beams must next be phase aligned. The initial offsets follow the phase-alignment rule in \eqref{eq:closed_form_phase_alignment_solution}. For the actual truncated sub-arrays used in the simulations, we refine the offsets by propagating each unaligned sub-array field to the target window without blockage and using the finite-window inner products in \eqref{eq:window_phase_gram}. For \(M=2\), the relative phase is obtained from \eqref{eq:two_beam_window_phase}; for \(M>2\), the coordinate update in \eqref{eq:window_phase_coordinate_update} is used. After phase alignment, the unnormalized multi-Airy beamforming vector is
	\begin{equation}
		\widetilde{\mathbf{w}}
		=
		\sum_{m=1}^{M}
		e^{\jmath\delta_m}
		\widetilde{\mathbf{w}}_m.
		\label{eq:multi_airy_unnormalized_weights}
	\end{equation}
	Finally, the beamforming vector is normalized under the total transmit-power constraint:
	\begin{equation}
		\mathbf{w}
		=
		\sqrt{P_{\mathrm{tot}}}
		\frac{\widetilde{\mathbf{w}}}
		{\|\widetilde{\mathbf{w}}\|_2},
		\label{eq:multi_airy_discrete_weights}
	\end{equation}
	Equations \eqref{eq:subarray_beamforming_field}--\eqref{eq:multi_airy_discrete_weights} constitute the proposed multi-Airy beamforming procedure.

	\begin{algorithm}[!t]
		\caption{Multi-Airy Beamforming}
		\label{alg:constructive_multi_airy}
		\footnotesize
	\begin{algorithmic}[1]
			\STATE \parbox[t]{0.9\linewidth}{\textbf{Input:} target position $(x_0,z_0)$, receiver window \(\mathcal{R}\), array center \(x_c\), antenna spacing \(d\), antenna number \(N\), Gaussian aperture width \(\omega_0\), specified sub-array number \(M\), specified trajectory-parameter set \(\boldsymbol{\eta}=\{(B_m,F_m,\theta_m)\}_{m=1}^{M}\), and total transmit-power limit \(P_{\mathrm{tot}}\)}
			\STATE Set \(x_n=x_c+(n-(N+1)/2)d\) and form the default uniform contiguous sub-array index sets \(\{\mathcal{N}_m\}_{m=1}^{M}\)
			\FOR{$m=1,\cdots,M$}
				\STATE Form $\widetilde{\mathbf{w}}_m$ using Eq. \eqref{eq:subarray_beamforming_field}
				\STATE Compute the unblocked target-window field \(u_m\) using Eq. \eqref{eq:actual_subarray_window_field}
			\ENDFOR
			\STATE Set \(\delta_1=0\); if \(M=2\), compute \(\delta_2\) from Eq. \eqref{eq:two_beam_window_phase}; if \(M>2\), compute \(\{\delta_m\}_{m=1}^{M}\) from Eq. \eqref{eq:window_phase_coordinate_update} with a prescribed sweep number or stopping tolerance
			\STATE Form $\widetilde{\mathbf{w}}$ using Eq. \eqref{eq:multi_airy_unnormalized_weights}
			\STATE Compute $\mathbf{w}$ using Eq. \eqref{eq:multi_airy_discrete_weights}
			\STATE \textbf{Output:} multi-Airy beamforming vector $\mathbf{w}$
	\end{algorithmic}
	\end{algorithm}

	\subsubsection{Intuitive Illustration of the Multi-Airy Working Mechanism}
	\begin{figure}
		\centering
		\includegraphics[width=0.90\linewidth]{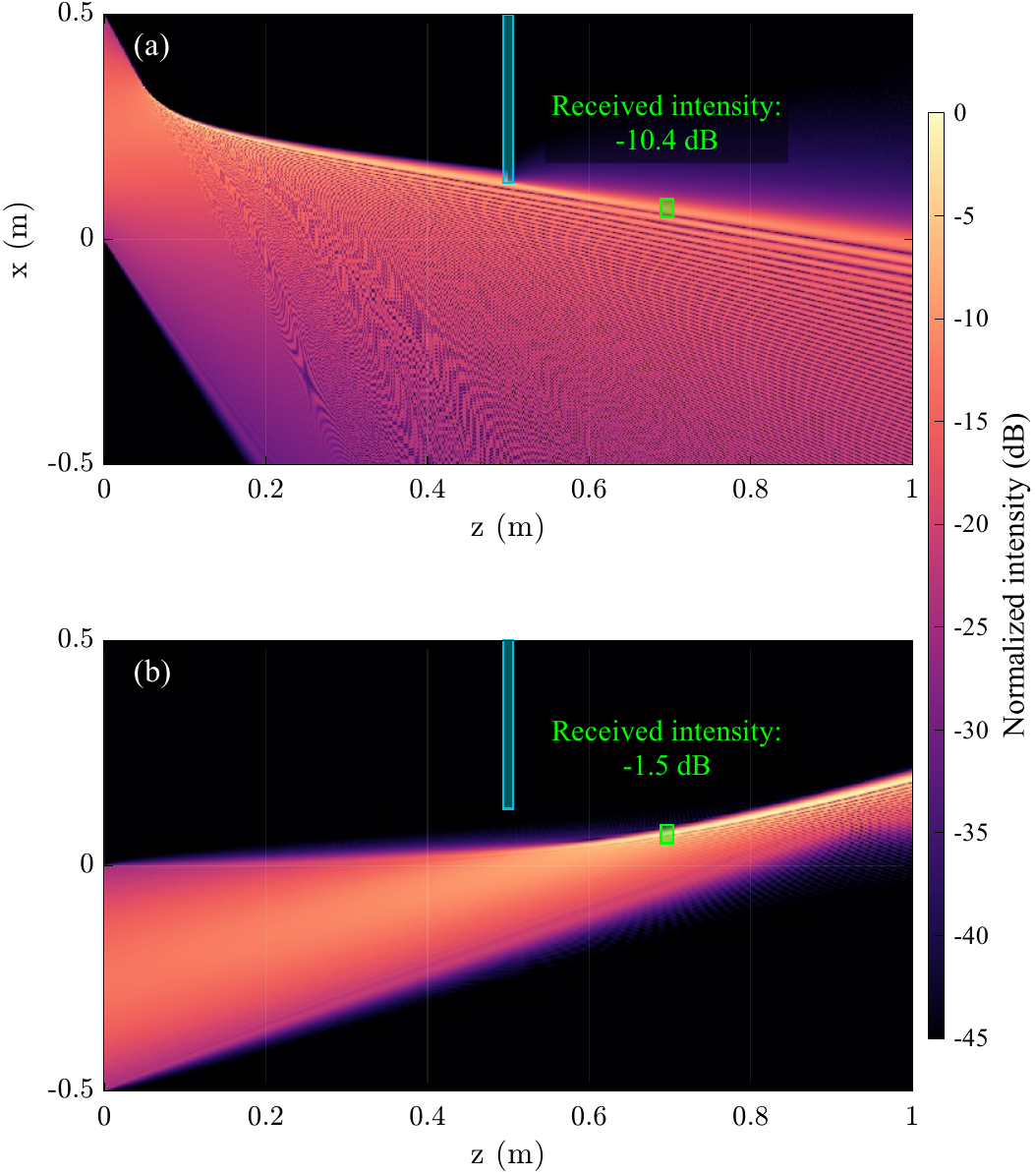}
	\caption{Representative Double-Airy example: (a) right sub-array beam; (b) left sub-array beam.}
		\label{fig:multi_airy_mechanism}
	\end{figure}

	To explain the working mechanism intuitively, Fig.~\ref{fig:multi_airy_mechanism} illustrates a representative example with \(M=2\), in which the ULA is divided into left and right half-arrays. Each half-array generates one Airy trajectory, and the two sub-array-generated trajectories experience different blockage conditions.

	In Fig.~\ref{fig:multi_airy_mechanism}(a), the right half-array faces more severe blockage and is therefore assigned trajectory parameters that generate a more strongly curved Airy trajectory outside the blocked region. In Fig.~\ref{fig:multi_airy_mechanism}(b), the left half-array faces lighter blockage and is assigned trajectory parameters that generate a less curved Airy trajectory with higher target-window average intensity. This example illustrates the central principle of the proposed scheme: different sub-arrays generate distinct Airy trajectories according to their respective propagation conditions. The resulting trajectory diversity reduces reliance on a single accurately placed trajectory, and the phase-alignment procedure in Section \ref{sec-multi-airy-theory} enables the sub-array contributions to be coherently combined at the target window.

	\subsubsection{Implementation Complexity}
	Given the specified \(M\) and trajectory-parameter set \(\boldsymbol{\eta}=\{(B_m,F_m,\theta_m)\}_{m=1}^{M}\), constructing the disjoint sub-array beamforming vectors in \eqref{eq:subarray_beamforming_field} requires \(\mathcal{O}(N)\) operations because every antenna belongs to only one sub-array. Using the point-wise phase-alignment rule in \eqref{eq:closed_form_phase_alignment_solution}, computing the \(M\) phase offsets costs \(\mathcal{O}(M)\), and coherent combination and total-power normalization require another \(\mathcal{O}(N)\) operations. Therefore, the constructive beamforming stage with known trajectory parameters has complexity \(\mathcal{O}(N+M)=\mathcal{O}(N)\), since \(M\leq N\). This complexity does not include the separate trajectory search or beam-training stage. If finite-window propagation-field refinement is used, forming the Gram matrix in \eqref{eq:window_phase_gram} adds the cost of propagating the \(M\) sub-array fields to the receiver window, and each coordinate sweep in \eqref{eq:window_phase_coordinate_update} costs \(\mathcal{O}(M^2)\) once \(\mathbf{Q}\) is available.

	\begin{figure*}[!t]
	\centering
	\includegraphics[width=0.78\textwidth]{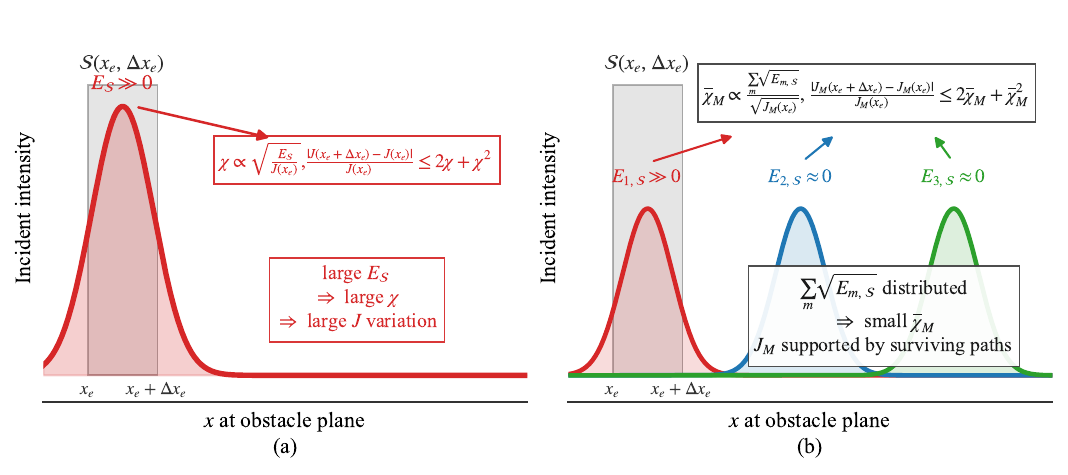}
	\caption{Interpretation of the geometry-sensitivity bounds: (a) single-Airy sensitivity; (b) multi-Airy robustness.}
	\label{fig:sensitivity_bound_interpretation}
	\end{figure*}

	\subsection{Performance Bound and Robustness Analysis} \label{subsec:robustness_analysis}
	After the Airy trajectory-parameter set \(\{(B_m,F_m,\theta_m)\}_{m=1}^{M}\) is selected by the beamforming procedure, the actual Tx-obs-Rx geometry may deviate from the nominal geometry because of estimation error or subsequent obstacle motion. In the planar-screen model, this deviation is represented by an obstacle-edge coordinate error, and we analyze how the received target-window signal energy changes under this geometry error.
	Let $x_e$ denote the estimated obstacle edge on the obstacle plane. The corresponding single-sided obstacle mask is \(M_{\mathrm{obs}}(x;x_e)\), as defined in Section \ref{sec-blockage-model}. For a fixed transmit field $\psi_0(x)$, define the field incident on the obstacle plane as
	\begin{equation}
		\psi_b(x)=\mathcal{D}_{z_{\mathrm{obs}}}\{\psi_0\}(x),
	\end{equation}
	where $\mathcal{D}_z$ is the free-space propagation operator defined in Section \ref{sec-blockage-model}. Using the blocked propagation operator in \eqref{eq:blockage_diffraction_operator}, the received field over the target window $\mathcal{R}$ under edge location $x_e$ is
	\begin{equation}
		u_{x_e}
		=\mathcal{P}_{\mathcal{R}}
		\mathcal{B}_{x_e,z_{\mathrm{obs}}}\{\psi_0\},
	\end{equation}
	where $\mathcal{P}_{\mathcal{R}}$ restricts the field to the receiver window. The corresponding target-window energy is
	\begin{equation}
		J(x_e)=\frac{1}{W_{\mathrm{rx}}}\left\|u_{x_e}\right\|_2^2.
	\end{equation}
	In this local parameterization, the perturbed edge is \(x_e+\Delta x_e\). The geometry error changes the obstacle mask only over the interval
	\begin{equation}
		\mathcal{S}(x_e,\Delta x_e)
		=
		\big[
		\min(x_e,x_e+\Delta x_e),
		\max(x_e,x_e+\Delta x_e)
		\big].
	\end{equation}
	Since \(z_0>z_{\mathrm{obs}}\) and the receiver window is finite, free-space diffraction from the obstacle plane to the target window is a bounded linear map. Let
	\begin{equation}
		C_{\mathcal{R}}(z_0-z_{\mathrm{obs}})
		=
		\left\|\mathcal{P}_{\mathcal{R}}\mathcal{D}_{z_0-z_{\mathrm{obs}}}\right\|_2
	\end{equation}
	denote the corresponding operator norm. The induced change in the received field then satisfies
	\begin{equation}
		\begin{aligned}
		&\left|
		\sqrt{J(x_e+\Delta x_e)}-\sqrt{J(x_e)}
		\right|  \\
		&\quad\le
		\frac{C_{\mathcal{R}}(z_0-z_{\mathrm{obs}})}{\sqrt{W_{\mathrm{rx}}}}
		\left(
		\int_{\mathcal{S}(x_e,\Delta x_e)}
		\left|\psi_b(x)\right|^2 \mathrm{d}x
		\right)^{1/2}.
		\end{aligned}
		\label{eq:edge_sqrt_bound}
	\end{equation}
	For a nonzero nominal target-window energy \(J(x_e)>0\), define the normalized geometry-sensitivity factor
	\begin{equation}
		\chi(x_e,\Delta x_e)=
		C_{\mathcal{R}}(z_0-z_{\mathrm{obs}})
		\left(
		\frac{
		\int_{\mathcal{S}(x_e,\Delta x_e)}
		\left|\psi_b(x)\right|^2 \mathrm{d}x
		}{
		W_{\mathrm{rx}}J(x_e)}
		\right)^{1/2},
	\end{equation}
	the relative target-energy variation is bounded by
	\begin{equation}
		\frac{|J(x_e+\Delta x_e)-J(x_e)|}{J(x_e)}
		\le
		2\chi(x_e,\Delta x_e)+\chi^2(x_e,\Delta x_e).
		\label{eq:edge_relative_bound}
	\end{equation}

	\textit{Single-Airy sensitivity:} For a fixed propagation geometry, \(C_{\mathcal{R}}(z_0-z_{\mathrm{obs}})\) is independent of the beamforming configuration. The sensitivity of a single Airy beam is therefore governed by the ratio between the incident field energy within the geometry-error strip and the nominal target-window energy that survives under the designed Tx-obs-Rx geometry. When the strip intersects the Airy main lobe, the numerator in \(\chi(x_e,\Delta x_e)\) can be large. If the nominal received signal energy \(J(x_e)\) is also mainly supported by the same dominant trajectory, the denominator does not provide a strong margin against this geometry error. Equation \eqref{eq:edge_relative_bound} thus identifies the regime in which a small Tx-obs-Rx geometry error can cause a large target-energy variation.

	\textit{Multi-Airy robustness:} For the multi-Airy beam, write the incident field as
	\begin{equation}
		\psi_b(x)
		=
		\sum_{m=1}^{M}e^{\jmath\delta_m}\psi_{b,m}(x),
	\end{equation}
	where \(\psi_{b,m}=\mathcal{D}_{z_{\mathrm{obs}}}\{\psi_{0,m}\}\). Let \(J_M\) and \(\chi_M\) denote \(J\) and \(\chi\), respectively, when evaluated for the coherently combined multi-Airy field. Define the energy of the \(m\)-th incident field within the geometry-error strip as
	\begin{equation}
		E_{m,\mathcal{S}}
		=
		\int_{\mathcal{S}(x_e,\Delta x_e)}
		|\psi_{b,m}(x)|^2\mathrm{d}x.
	\end{equation}
	Applying the triangle inequality to the error-induced field difference gives
	\begin{equation}
		\begin{aligned}
			\chi_M(x_e,\Delta x_e)
			&\leq \overline{\chi}_M,\\
			\overline{\chi}_M
			&\triangleq
			\frac{C_{\mathcal{R}}(z_0-z_{\mathrm{obs}})}
			{\sqrt{W_{\mathrm{rx}}J_M(x_e)}}
			\sum_{m=1}^{M}\sqrt{E_{m,\mathcal{S}}}.
		\end{aligned}
		\label{eq:multi_edge_sensitivity}
	\end{equation}
	where \(J_M(x_e)\) is the target-window energy of the coherently combined multi-Airy beam. Combining \eqref{eq:edge_relative_bound} and \eqref{eq:multi_edge_sensitivity} yields
	\begin{equation}
		\frac{|J_M(x_e+\Delta x_e)-J_M(x_e)|}{J_M(x_e)}
		\leq
		2\overline{\chi}_M+\overline{\chi}_M^2.
	\end{equation}

	Fig.~\ref{fig:sensitivity_bound_interpretation} illustrates the physical meaning of \eqref{eq:edge_relative_bound} and \eqref{eq:multi_edge_sensitivity}. For a single Airy beam, the geometry-error strip \(\mathcal{S}(x_e,\Delta x_e)\) can overlap the only dominant main lobe, making the numerator in \(\chi(x_e,\Delta x_e)\) large. As a result, the relative variation of \(J(x_e)\) can be large, which is reflected in a significant attenuation of the target-position energy. For multi-Airy beamforming, the same local strip first affects only part of the spatially separated trajectories. Therefore, only the corresponding \(E_{m,\mathcal{S}}\) terms become large, and the surviving trajectories can still support \(J_M(x_e)\), so the target-position energy retains stronger robustness.
	
	Equation \eqref{eq:multi_edge_sensitivity} further explains the robustness mechanism of multi-Airy beamforming. When the geometry-error strip affects only a subset of the spatially separated trajectories, only the corresponding \(E_{m,\mathcal{S}}\) terms become large, and \(J_M(x_e)\) can still be maintained by the surviving trajectories. Consequently, the relative variation of \(J_M(x_e)\) need not be large under a localized Tx-obs-Rx geometry error. Of course, this result does not imply unconditional dominance over single-Airy beamforming. The multi-Airy advantage can diminish if the obstacle simultaneously blocks all trajectories, if the trajectories overlap strongly at the obstacle plane, or if sub-array partition reduces \(J_M(x_e)\) excessively. Together, the constructive rule and the robustness bounds explain why trajectory diversity can reduce sensitivity to Tx-obs-Rx geometry errors.
	
	\section{Simulation Results and Discussions} \label{sec-simulation}
	This section evaluates the proposed multi-Airy beamforming scheme through numerical simulations. We first specify the simulation parameters and rate metric. We then compare full-space intensity distributions, evaluate robustness to Tx-obs-Rx geometry estimation errors, compare the achievable rates of single-Airy and multi-Airy beamforming, and provide a UPA receive-plane validation.

\subsection{Simulation Setup}

We consider a near-field uniform linear array (ULA) with single-sided blockage. Unless otherwise stated, the target is located on broadside at \(z_0 = 5\) m, and the obstacle plane is placed at \(z_{\mathrm{obs}} = 4.5\) m. The carrier frequencies are \(f_c = 100\) GHz and \(30\) GHz, where the latter is included for millimeter-wave comparison. The aperture length is \(L = 1.0\) m, the bandwidth is \(B_{\mathrm w} = 5\) GHz for both carrier frequencies, and the normalized noise power is \(N_0 = 3.16 \times 10^{-2}\). A receiver window of width \(W_{\mathrm{rx}} = 6\) mm is centered at the target.

Six beamforming schemes are evaluated: focused uniform beam, focused Gaussian beam, single-Airy (\(M=1\)), double-Airy (\(M=2\)), triple-Airy (\(M=3\)), and four-Airy (\(M=4\)). The focused uniform beam applies equal-amplitude weights with near-field phase focusing, while the focused Gaussian beam applies a Gaussian amplitude taper with identical near-field phase focusing. All schemes are normalized to the same total transmit power \(P_{\mathrm{tot}} = 1\) in normalized units, consistent with \eqref{eq:multi_airy_design_problem}.

\begin{table}[!t]
	\centering
	\caption{Main Simulation Parameters}
	\label{tab:simulation_parameters}
	\begin{tabular}{@{}p{0.58\linewidth}p{0.32\linewidth}@{}}
		\toprule
		Parameter & Value \\
		\midrule
		Carrier frequency \(f_c\) & \(100\) GHz / \(30\) GHz \\
		Aperture length \(L\) & \(1.0\) m \\
		Target location & \((x_0,z_0)=(0,5~\mathrm{m})\) \\
		Receiver-window width \(W_{\mathrm{rx}}\) & \(6\) mm \\
		Bandwidth \(B_{\mathrm w}\) & \(5\) GHz \\
		Normalized noise power \(N_0\) & \(3.16\times 10^{-2}\) \\
		Normalized total transmit power \(P_{\mathrm{tot}}\) & \(1\) \\
		\bottomrule
	\end{tabular}
\end{table}

To bridge the field simulation with communication metrics, the effective received intensity is defined as the average intensity over the target window,

\begin{equation}
	J_{\mathrm{rx}}=\frac{1}{W_{\mathrm{rx}}}
	\int_{x_0-W_{\mathrm{rx}}/2}^{x_0+W_{\mathrm{rx}}/2}
	\left| \psi_{\mathrm{tot}}(x,z_0) \right|^2 \mathrm{d}x,
\end{equation}
where \(\psi_{\mathrm{tot}}(x,z_0)\) is the target-plane field after blockage. With the common transmit-power normalization, \(J_{\mathrm{rx}}/N_0\) serves as the normalized receive SNR for rate comparison. The corresponding achievable rate is

\begin{equation}
	R = B_{\mathrm w} \log_2 \left( 1 + \frac{J_{\mathrm{rx}}}{N_0} \right).
\end{equation}

\begin{figure*}[!t]
	\centering
	\includegraphics[width=0.86\textwidth]{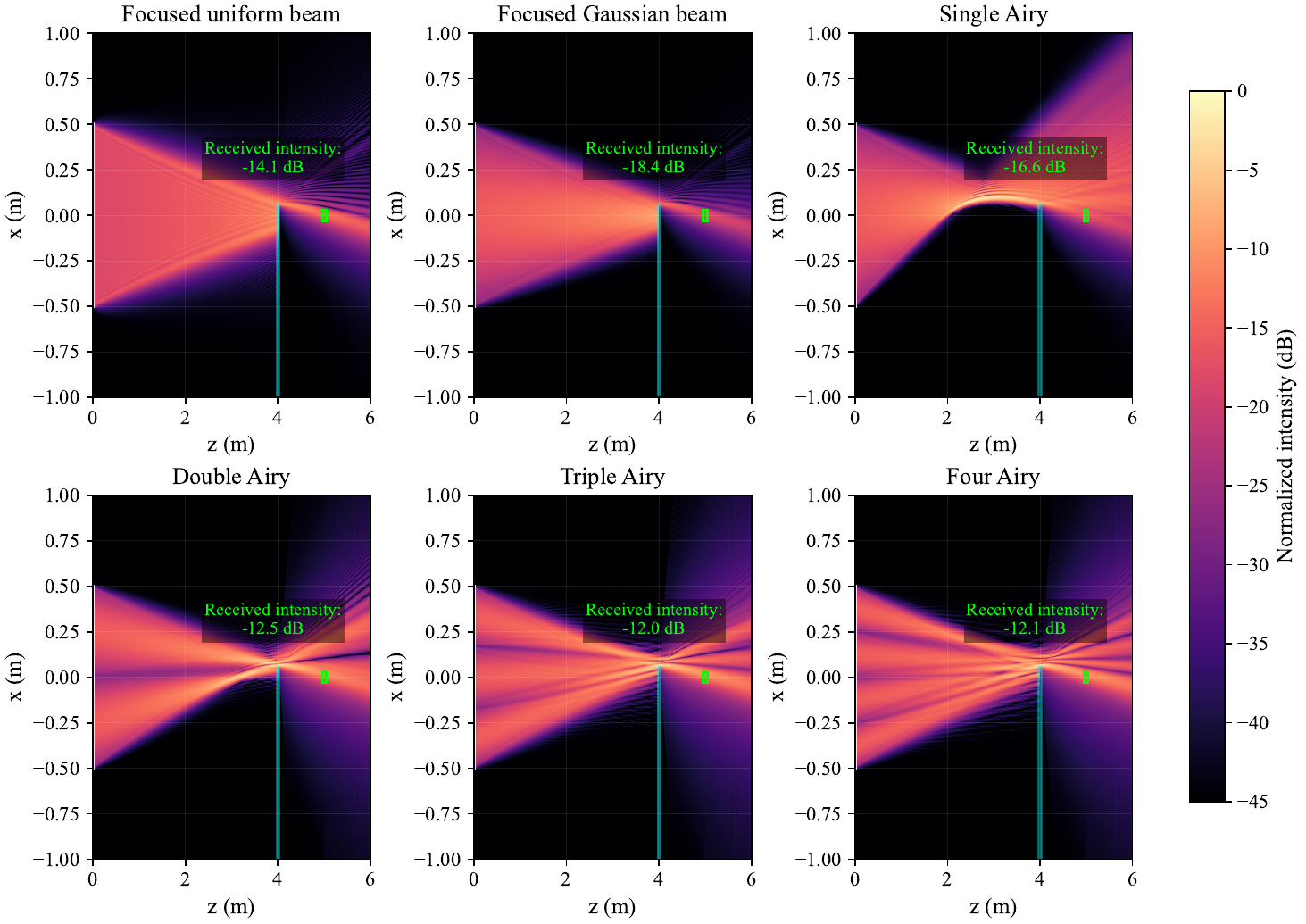}
	\caption{Normalized full-space intensity distributions at \(\mathrm{bl}=0.8\).}
	\label{fig:blocked_intensity}
\end{figure*}

\subsection{Comparison of Full-Space Intensity Distributions}

Fig.~\ref{fig:blocked_intensity} compares the full-space intensity distributions at \(\mathrm{bl}=0.8\) with \(z_{\mathrm{obs}} = 4\) m. The focused uniform beam and focused Gaussian beam achieve target-window average intensities of \(-14.1\) dB and \(-18.4\) dB, respectively. The single-Airy beam concentrates its energy along a single curved trajectory, yielding \(-16.6\) dB after blockage. In contrast, the double-Airy, triple-Airy, and four-Airy schemes generate multiple curved trajectories and achieve \(-12.5\) dB, \(-12.0\) dB, and \(-12.1\) dB, respectively. These results demonstrate that multi-Airy beamforming, by exploiting path diversity, attains higher target-window energy than single-Airy beamforming under strong blockage.

\subsection{Robustness to Tx-Obs-Rx Geometry Estimation Errors}

For this evaluation, we use the same aperture and link distance, with \(\theta = 0\) and \(z_{\mathrm{obs}} = 4.5\) m, and a nominal invisible ratio of \(\mathrm{bl} = 0.6\). The trajectory parameters \((B, F, \theta)\) of each Airy scheme are selected for the estimated geometry and then held fixed as the actual obstacle intrudes further into the nominally visible region, testing the sensitivity mechanism analyzed in Section~\ref{subsec:robustness_analysis}.

\begin{figure}
	\centering
	\includegraphics[width=0.9\linewidth]{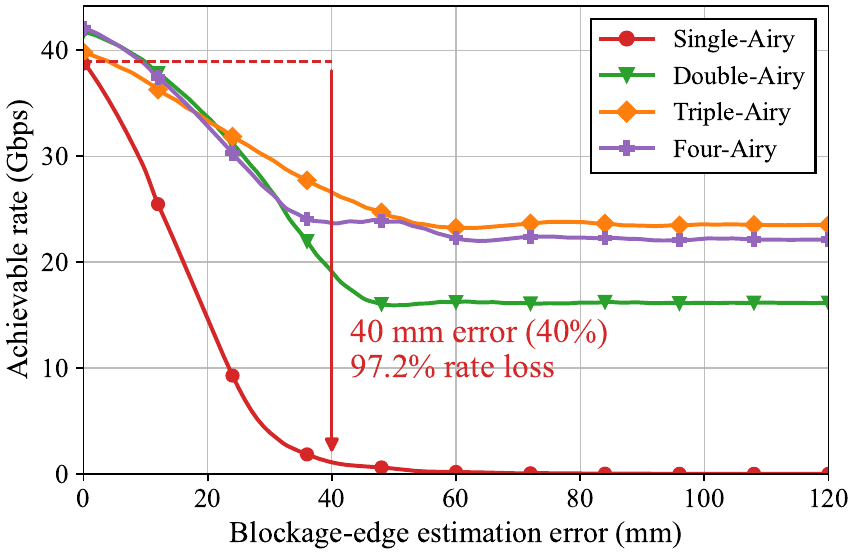}
	\caption{Achievable rate versus Tx-obs-Rx geometry estimation error.}
	\label{fig:edge_error_robustness}
\end{figure}

Fig.~\ref{fig:edge_error_robustness} shows the achievable rate when the actual geometry deviates from the nominal one used for beam selection. In this planar-screen model, a positive error indicates that the obstacle extends farther into the unblocked region in the estimation. All Airy parameters remain fixed throughout the error sweep.

The single-Airy rate drops from \(38.91\) Gbps at zero error to \(1.09\) Gbps at a \(40\)-mm edge-coordinate error, causing a \(97.2\%\) loss in achievable rate. This centimeter-level error is much smaller than the 0.75-m horizontal positioning-accuracy requirement specified for the Indoor Factory-ISAC test environment in IMT-2030~\cite{ma2026imt2030tpr}, yet it is sufficient to move the obstacle edge across the narrow high-intensity region of the designed Airy trajectory. At the same error, the best multi-Airy configuration limits the achievable-rate loss to \(33.3\%\), and all considered multi-Airy schemes keep the loss below \(54.3\%\). This result complements the received signal energy example with a smaller error in Fig.~\ref{fig:single_airy_edge_error_field}, corroborates the robustness mechanism based on trajectory diversity described by \eqref{eq:multi_edge_sensitivity}, and confirms that the multi-Airy robustness advantage becomes particularly pronounced under geometry estimation errors.

\subsection{Communication Performance Under Blockage}

\begin{figure}
	\centering
	\includegraphics[width=0.90\linewidth]{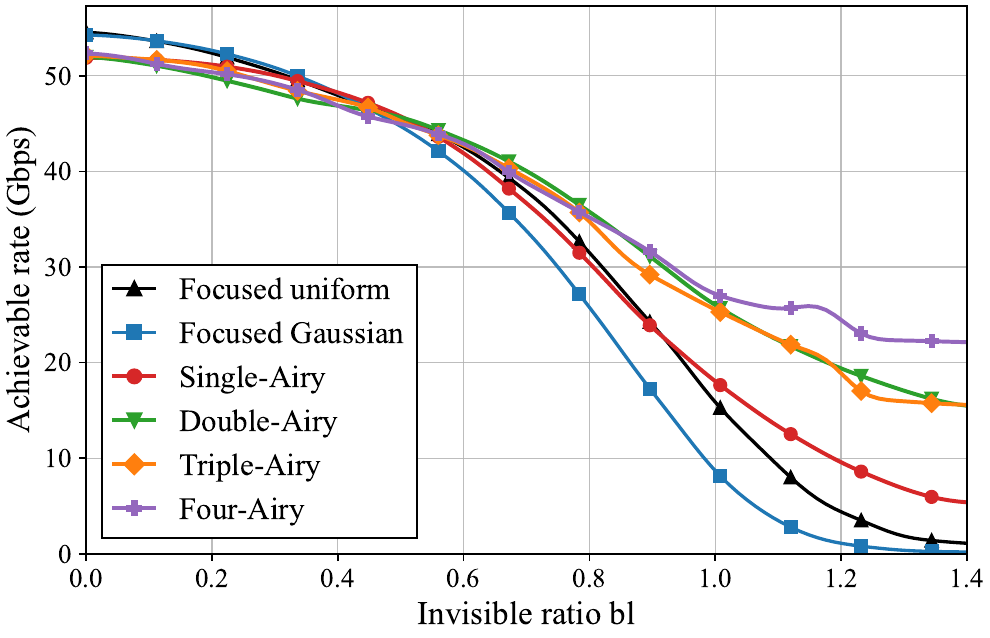}
	\caption{Achievable rate versus invisible ratio.}
	\label{fig:full_range}
\end{figure}

Fig.~\ref{fig:full_range} plots the achievable rate versus the invisible ratio over \(0 \leq \mathrm{bl} \leq 1.4\). When \(\mathrm{bl} > 1\), the obstacle edge is extrapolated beyond complete blockage of the target-visible aperture span, so no direct LoS aperture remains. In this situation, the energy of the target window is primarily propagated by the diffracted field, and the multi-Airy beamforming scheme consistently provides an obvious rate advantage by exploiting its path diversity. For example, at \(\mathrm{bl} = 0.8\), the four-Airy scheme achieves \(35.27\) Gbps, compared with \(30.55\) Gbps for single-Airy beamforming and \(31.65\) Gbps for the focused uniform beam. At \(\mathrm{bl} = 1.2\), the four-Airy scheme maintains \(24.27\) Gbps, whereas the single-Airy and focused uniform beams drop to \(9.45\) Gbps and \(4.27\) Gbps, respectively.

\begin{figure}
	\centering
	\includegraphics[width=0.90\linewidth]{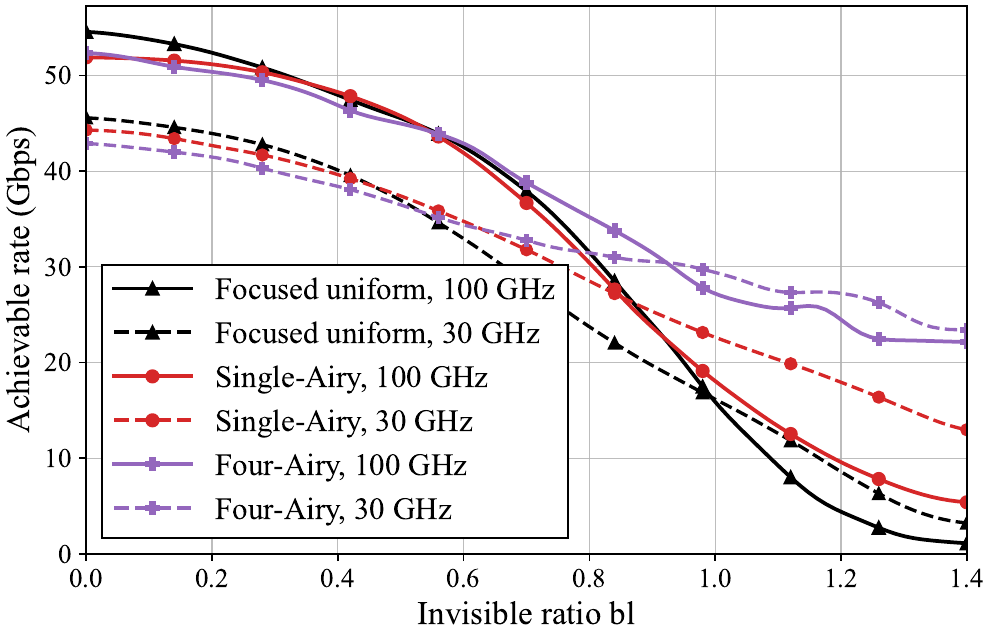}
	\caption{Achievable rate versus invisible ratio at \(100\) GHz and \(30\) GHz.}
	\label{fig:frequency_comparison}
\end{figure}

To examine the proposed scheme's effectiveness in the millimeter-wave band, Fig.~\ref{fig:frequency_comparison} compares the focused uniform beam, single-Airy, and four-Airy beamforming at \(100\) GHz and \(30\) GHz. The same aperture, total transmit power, obstacle geometry, bandwidth, normalized noise power, and trajectory-selection rule are used at both frequencies. At \(\mathrm{bl} = 0.8\), four-Airy beamforming achieves \(35.27\) Gbps and \(31.47\) Gbps at \(100\) GHz and \(30\) GHz, respectively, compared with \(30.55\) Gbps and \(28.57\) Gbps for single-Airy beamforming. At \(\mathrm{bl} = 1.2\), four-Airy beamforming maintains \(24.27\) Gbps and \(27.19\) Gbps, respectively. These millimeter-wave results preserve the same strong-blockage robustness trend observed at sub-THz frequencies.

\begin{figure}
	\centering
	\includegraphics[width=0.96\linewidth]{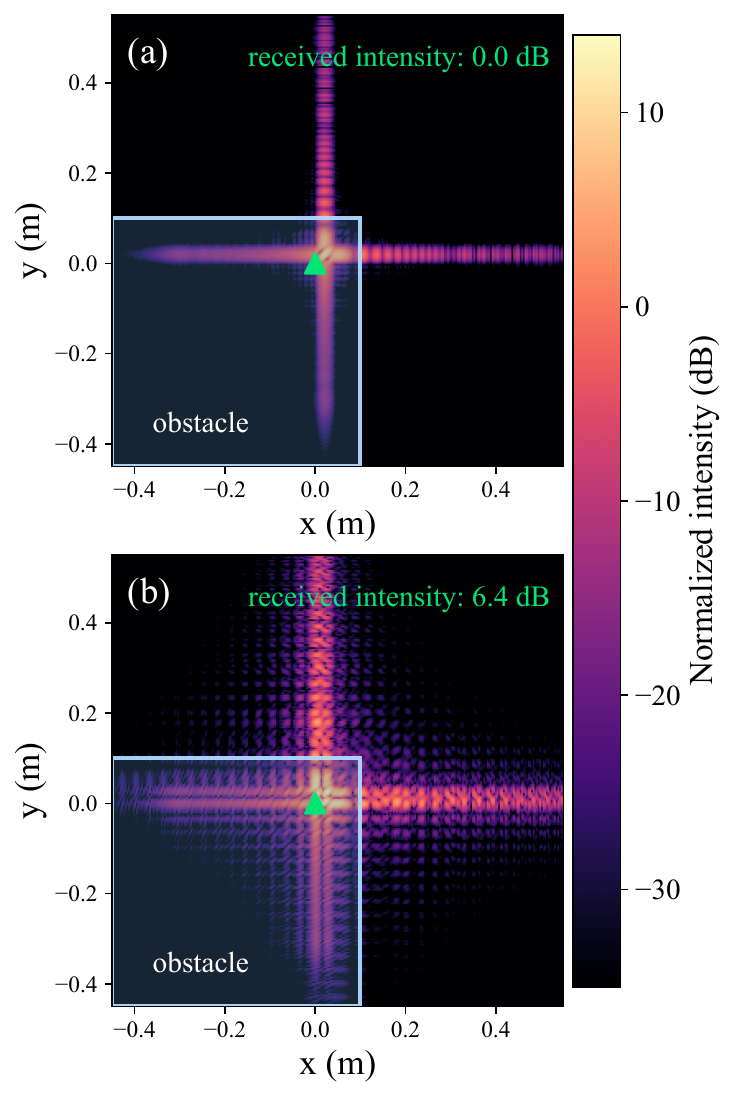}
	\caption{Receive-plane intensity comparison for UPA Airy beamforming under rectangular blockage: (a) single-Airy beamforming; (b) four-Airy beamforming. The blue rectangles indicate the obstacle region.}
	\label{fig:upa_receive_plane}
\end{figure}

\subsection{UPA Receive-Plane Validation}

For this validation, the UPA Airy field is generated in a separable form along the \(x\)- and \(y\)-directions. In the four-Airy case, the square aperture is divided into four quadrant sub-arrays, and each quadrant configures a tailored separable Airy field. The constant phase offsets are then obtained by aligning the component responses at the target, following the same coherent-combining principle developed for the ULA case.

Fig.~\ref{fig:upa_receive_plane} further validates the multi-Airy principle on a \(1~\mathrm{m}\times1~\mathrm{m}\) UPA. The target is located at \((x_0,y_0,z_0)=(0,0,5~\mathrm{m})\), and a rectangular obstacle is placed at \(z_{\mathrm{obs}}=4\) m with \(x\in[-1,0.1]\) m and \(y\in[-1,0.1]\) m. The blue rectangle marks the obstacle region, and the four-Airy case partitions the UPA into four quadrant sub-arrays. The received intensity is averaged over a \(0.04~\mathrm{m}\times0.04~\mathrm{m}\) window centered at the target. Under the same total transmit-power normalization, the received-window average intensity of single-Airy beamforming is used as the \(0\)-dB reference, and four-Airy beamforming achieves a \(6.41\)-dB improvement. With the additional transverse dimension of the UPA, the path-diversity gain of multi-Airy beamforming becomes more pronounced. This result confirms that the trajectory-diversity and coherent-combining principle remains effective for planar arrays.

\vspace{-1.0em}
\section{Conclusions}
	
This paper addressed the Tx-obs-Rx geometry sensitivity of single-Airy beamforming by developing a multi-Airy beamforming scheme with coordinated curved trajectories. The analysis showed that a small Tx-obs-Rx geometry error can attenuate the dominant single-Airy main lobe contribution, while multi-Airy generation mitigates this problem through trajectory diversity and coherent combining. In the robustness test, under a 40-mm edge-coordinate error, the single-Airy beamforming scheme retained only 2.8\% of the rate achieved under perfectly estimated geometry, corresponding to a \(97.2\%\) loss in achievable rate. In sharp contrast, the best multi-Airy configuration reduces this loss to \(33.3\%\), about 1/3 of that incurred by the single-Airy scheme under the same error. In the blockage-rate sweep, four-Airy beamforming achieved 2.57x and 5.68x rate improvements over the single-Airy beamforming and the focused uniform beamforming at \(\mathrm{bl}=1.2\), respectively. Beyond these performance gains, this work provides a new perspective on blockage mitigation by, for the first time, explicitly quantifying the sensitivity problem of single-Airy beamforming. Moreover, the proposed coordinated multi-Airy generation method is general and inherently extensible to distributed systems, as it does not rely on specific sub-array configurations but accommodates arbitrary array geometries. Future work can extend the proposed framework to multi-user communication, wideband transmission, beam training, and dynamic beam tracking scenarios.

\appendices
\section{Proof of Theorem~\ref{thm:phase_alignment_condition}}
\label{app:phase_alignment_proof}
\begin{IEEEproof}
	By the triangle inequality, the combined target-response magnitude satisfies
	\begin{equation}
		\left|
		\sum_{m=1}^{M}e^{\jmath\delta_m}g_m
		\right|
		\leq
		\sum_{m=1}^{M}|g_m|.
	\end{equation}
	The right-hand side is an upper bound on the target-response magnitude for the specified individual Airy beam responses \(\{g_m\}\). Writing \(g_m=|g_m|e^{\jmath\varphi_m}\), the phase dependence can be expressed as
	\begin{equation}
		\begin{aligned}
		\left|
		\sum_{m=1}^{M}e^{\jmath\delta_m}g_m
		\right|^2
		&=
		\sum_{m=1}^{M}|g_m|^2\\
		&\quad+
		2\sum_{m<\ell}|g_m||g_\ell|
		\cos\!\left(
		\delta_m+\varphi_m-\delta_\ell-\varphi_\ell
		\right).
		\end{aligned}
	\end{equation}
	Therefore, phase misalignment reduces the coherent cross terms, and the upper bound is attained when all nonzero terms \(e^{\jmath\delta_m}g_m\) have the same phase. Choosing \(\delta_m^{(0)}=\phi_0-\varphi_m\) gives \(e^{\jmath\delta_m^{(0)}}g_m=e^{\jmath\phi_0}|g_m|\) for every nonzero \(g_m\), which attains the maximum target-response power \(\left(\sum_{m=1}^{M}|g_m|\right)^2\).
\end{IEEEproof}

	\footnotesize
	\bibliographystyle{IEEEtran} 
	\bibliography{ref}  
	\normalsize
	
\end{document}